\documentclass[12pt,journal,onecolumn]{./IEEEtran}
\usepackage[english]{babel}
\usepackage{amsmath}
\usepackage{amsfonts}
\usepackage{amssymb}
\usepackage{graphics}

\renewcommand{\baselinestretch}{1.66}
\usepackage[left=1in,top=1in,right=1in,bottom=1in]{geometry}

\usepackage{graphicx,color}
\def\href#1#2{{\tt #2}}

\def\SHAPE#1{\mbox{{\tt #1}}}
\def\suff{\mathit{suff}}

\newtheorem{theorem}{Theorem}
\newtheorem{lemma}{Lemma}
\newtheorem{example}{Example}

\begin{document}

\title{Multiseed lossless filtration
\thanks{An extended abstract of
    this work has been presented to the {\em Combinatorial Pattern
      Matching} conference (Istanbul, July 2004)}
}

\author{Gregory Kucherov, Laurent No{\'e}, Mikhail Roytberg
\thanks{Gregory Kucherov and Laurent No{\'e} are with the INRIA/LORIA, 
615, rue du Jardin Botanique, B.P. 101, 54602 
Villers-l\`es-Nancy, France, e-mail: {\tt [Gregory.Kucherov,Laurent.Noe]@loria.fr}}
\thanks{Mikhail Roytberg is with the 
Institute of Mathematical Problems in Biology,
Pushchino, Moscow Region, Russia, e-mail: {\tt roytberg@impb.psn.ru}}
}

\maketitle

\begin{abstract}
We study a method of seed-based lossless filtration for approximate string
matching and related bioinformatics applications. The method is based on a
simultaneous use of several spaced seeds rather
than a single seed as studied by Burkhardt and
Karkkainen~\cite{burkhardt03better}. We present algorithms to
compute several important parameters of seed families, study 
their combinatorial properties, and describe several techniques to
construct efficient families. We also report a large-scale application
of the proposed technique to the problem of oligonucleotide selection
for an EST sequence database.
\end{abstract}

\begin{section}{Introduction}

\PARstart{F}{iltering} is a widely-used technique in 
biosequence analysis.
Applied to the approximate string matching problem \cite{FlexiblePatternMatching02}, it
can be summarized by the following two-stage scheme: to find approximate 
occurrences (matches) of a given string in a sequence (text), one first quickly
discards (filters out) those sequence regions where matches cannot occur, and then checks out the remaining parts of the
sequence for actual matches. The filtering is done according to
small patterns of a specified form that the searched string is assumed to
share, in the exact way, with its approximate occurrences. 
A similar filtration scheme is used by heuristic local alignment
algorithms (\cite{GBLAST97,PatternHunter02,BLASTZ03,NoeKucherovBMCBioinformatics04}, to
mention a few): they
first identify potential similarity regions that share some patterns
and then actually check whether those regions represent a significant
similarity by computing a corresponding alignment. 

Two types of filtering should be distinguished -- {\em lossless} and
{\em lossy}. A lossless filtration guarantees to detect {\em all}
sequence fragments under interest, while a lossy filtration may miss some
of them, but still tries to detect a majority of them. Local
alignment algorithms usually use a lossy filtration. On the other
hand, the lossless filtration has been studied in the context of
approximate string matching problem
\cite{PW95,burkhardt03better}. In this paper, we focus on
the lossless filtration.

In the case of lossy filtration, its efficiency is measured by two
parameters, usually called {\em selectivity} and 
{\em sensitivity}. The sensitivity measures the part of sequence fragments
of interest that are missed by the filter (false negatives), and the
selectivity indicates what part of detected candidate fragments don't actually
represent a solution (false positives). In the case of lossless
filtration, only the selectivity parameter makes sense and is
therefore the main characteristic of the filtration efficiency. 

The choice of patterns that must be contained in the searched sequence
fragments is a key ingredient of the filtration algorithm. 
{\em Gapped seeds} (spaced seeds, gapped $q$-grams) have been recently
shown to significantly improve the filtration efficiency over the
``traditional'' technique of contiguous seeds. In the framework of lossy
filtration for sequence alignment, the use of designed gapped seeds has been
introduced by the {\sc PatternHunter} method \cite{PatternHunter02} and
then used by some other algorithms (e.g. \cite{BLASTZ03,NoeKucherovBMCBioinformatics04}). In 
\cite{FLASH93,Buhler02}, spaced seeds have been shown to improve 
indexing schemes for similarity search in sequence
databases. The estimation of the sensitivity of spaced seeds (as
well as of some extended seed models) has been subject of several recent
studies \cite{KLMT04,BKS03,BBVb03,KucherovNoePontyBIBE04,ChoiZhang03,MiklosCPM04}. In the
framework of lossless filtration for approximate pattern matching,
gapped seeds were studied in \cite{burkhardt03better} (see also
\cite{PW95}) and have also been shown to increase the
filtration efficiency considerably. 

In this paper, we study an extension of the lossless single-seed filtration
technique~\cite{burkhardt03better}. The extension is
based on using {\em seed families} rather than individual seeds. 
The idea of simultaneous use of multiple seeds for DNA local
alignment 
was already 
envisaged
in \cite{PatternHunter02} and applied in {\sc PatternHunter~II} software
\cite{PatternHunter04}. The problem of designing efficient seed
families has also been studied in \cite{SB04}. In \cite{Brown2004},
multiple seeds have been applied to the protein search. 
However, the issues analysed in the present paper are quite different,
due to the proposed requirement for the search to be lossless.

The rest of the paper is organized as follows. After formally
introducing the concept of multiple seed filtering in
Section~\ref{msfilter}, Section~\ref{dyn-progr} is devoted to dynamic programming
algorithms to compute several important parameters of seed
families. In Section~\ref{seed-design}, we first study several
combinatorial properties of families of seeds, and, in particular,
seeds having a periodic structure. These results are used to obtain a
method for constructing efficient seed families. We also outline a
heuristic genetic programming algorithm for constructing seed
families. Finally, in Section~\ref{experiments}, we present several
seed families we computed, and we report a large-scale experimental
application of the method to a practical problem of oligonucleotide
selection.

\end{section}

\begin{section}{Multiple seed filtering}
\label{msfilter}
  
A {\em seed} $Q$ (called also {\em spaced seed} or 
{\em gapped $q$-gram}) is a list $\{p_1,p_2,\ldots,p_d\}$ of positive
integers, called {\em matching positions}, such that $p_1<p_2<\ldots<p_d$. 
By convention, we always assume $p_1=0$. The {\em span} of a seed $Q$,
denoted $s(Q)$, is the quantity  $p_d+1$. The number $d$ of
matching positions is called the {\em weight} of the seed and denoted
$w(Q)$.
Often we will use a more visual representation of seeds, adopted in
\cite{burkhardt03better}, as words of length $s(Q)$ over the two-letter alphabet
\{\SHAPE{\#},\SHAPE{-}\}, where \SHAPE{\#} occurs at all
matching
positions and \SHAPE{-} at all positions in between. For example, seed
$\{0,1,2,4,6,9,10,11\}$ of weight $8$ and span $12$ is represented by
word \SHAPE{\#\#\#-\#-\#--\#\#\#}. The character \SHAPE{-} is
called a {\em joker}.
Note that, unless otherwise stated, the seed has the character \SHAPE{\#} at its first and last
positions. 

Intuitively, a seed specifies the set of patterns that, if shared by two 
sequences, indicate a possible similarity between them. Two 
sequences are similar if the Hamming distance between them is smaller
than a certain threshold. For example, sequences {\tt CACTCGT} and {\tt CACACTT}
are similar within Hamming distance 2 and this similarity is detected by the
seed \SHAPE{\#\#-\#} at position 2. We are interested in seeds that
detect {\em all} similarities of a given length with a given Hamming
distance. 

Formally, a {\em gapless similarity} (hereafter simply {\em similarity}) of two
sequences of length $m$ is a binary word 
$w\in\{0,1\}^m$
interpreted as a sequence
of matches ($1$'s) and mismatches ($0$'s) of individual characters from the
alphabet of input sequences.
A seed $Q=\{p_1,p_2,\ldots,p_d\}$ {\em matches} a similarity $w$ at position $i$,
$1\leq i \leq m-p_d+1$, iff for every $j \in [1..d]$, we have $w[i+p_j]=1$. 
In this case, we also say that seed $Q$ {\em has an occurrence} in
similarity $w$ at position $i$. A seed $Q$ is said to 
{\em detect a similarity} $w$ if $Q$ has at least one occurrence in
$w$. 

Given a similarity length $m$ and a number of mismatches $k$, consider all
similarities of length $m$ containing $k$ $0$'s and $(m-k)$
$1$'s. These similarities are called $(m,k)$-similarities. A seed $Q$
{\em solves the detection problem $(m,k)$} (for short, the
$(m,k)$-problem) iff all of ${m \choose k}$ $(m,k)$-similarities $w$
are detected by $Q$.
For example, one can check that seed \SHAPE{\#-\#\#--\#-\#\#} solves
the $(15,2)$-problem.

Note that the weight of the seed is directly related to the
{\em selectivity} of the corresponding filtration procedure. A larger
weight improves the selectivity, as less similarities will pass
through the filter. On the other hand, a smaller weight reduces the filtration 
efficiency. Therefore, the goal is to solve an $(m,k)$-problem by a 
seed with the largest possible weight.

Solving $(m,k)$-problems by a single seed has been studied by
Burkhardt and K\"arkk\"ainen \cite{burkhardt03better}. 
An extension we propose here is to use a {\em family of seeds},
instead of a single seed, to solve the
$(m,k)$-problem. 
Formally, a finite family of seeds $F=<Q_l>_{l=1}^{L}$
{\em solves an $(m,k)$-problem} iff for % all ${m \choose k}$
any $(m,k)$-similarity $w$, there exists a seed $Q_l\in F$ that detects
$w$. 

Note that the seeds of the family are used in the
complementary (or disjunctive) fashion, i.e. a similarity is detected
if it is detected by {\em one of the seeds}. This differs
from the conjunctive approach of \cite{PW95} where a similarity should
be detected by two seeds {\em simultaneously}.

The following example motivates the use of multiple seeds. In
 \cite{burkhardt03better}, it has been shown that a seed solving the
 $(25,2)$-problem has the maximal weight 12.
The only such seed (up to reversal) is
 \SHAPE{\#\#\#-\#--\#\#\#-\#--\#\#\#-\#}.
 However, the problem can be solved by the family composed of the
 following two seeds of weight 14: \SHAPE{\#\#\#\#\#-\#\#---\#\#\#\#\#-\#\#} and
 \SHAPE{\#-\#\#---\#\#\#\#\#-\#\#---\#\#\#\#}.

Clearly, using
these two seeds increases the selectivity of the search, as only 
similarities having 14 or more matching characters pass the filter vs
12 matching characters in the case of single seed. On uniform
Bernoulli sequences, this results in the decrease of the number of
candidate similarities by the factor of $|A|^2/2$, where $A$ is the
input alphabet. This illustrates
the advantage of the multiple seed approach: it allows to increase the
selectivity while preserving a lossless search. The price to pay for
this gain in selectivity is multiplying the work on identifying the
seed occurrences. In the case of large sequences, however, this is largely compensated by the 
decrease in the number of false positives caused by the increase of
the seed weight.

\end{section}

\begin{section}{Computing properties of seed families}
\label{dyn-progr}

Burkhardt and  K\"arkk\"ainen  \cite{burkhardt03better} proposed a dynamic
programming algorithm to compute the {\em optimal threshold} of a
given seed 
-- the minimal number of its occurrences over all possible $(m,k)$-similarities.
In this section, we describe an extension of this algorithm for seed 
families and, on the other hand, describe dynamic programming
algorithms for computing two other important parameters of seed
families that we will use in a later section.

Consider an $(m,k)$-problem and a family of seeds
$F=<Q_l>_{l=1}^{L}$. We need the following notation. 
\begin{itemize}  
\item $s_{max} =max \{s( Q_l )\}_{l=1}^{L}$, $s_{min} =min \{s( Q_l
  )\}_{l=1}^{L}$, 
\item for a binary word $w$ and a seed $Q_l$, $\suff(Q_l,w)\!=\!1$ if $Q_l$
  matches $w$ at position $(|w|\!-\!s(Q_l)\!+\!1)$ (i.e. matches a suffix of
  $w$), otherwise $\suff(Q_l,w)\!=\!0$, 
\item $last(w)=1$ if the last character of $w$ is $1$, otherwise $last(w)=0$,
\item $zeros(w)$ is the number of $0$'s in $w$.
\end{itemize}

\begin{subsection}{Optimal threshold}

Given an $(m,k)$-problem, a family of seeds $F=<Q_l>_{l=1}^{L}$ has
the {\em optimal threshold} $T_F(m,k)$ if every $(m,k)$-similarity has
at least $T_F(m,k)$ occurrences of seeds of $F$ and this is the
maximal number with this property. 
Note that overlapping occurrences
of a seed as well as occurrences of different seeds at the same
position are counted separately. For example, the singleton family
\{\SHAPE{\#\#\#-\#\#}\} has threshold 2 for the $(15,2)$-problem. 

Clearly, $F$ solves an $(m,k)$-problem if and only if $T_F(m,k)>0$. If
$T_F(m,k)>1$, then one can strengthen the detection criterion by
requiring several seed occurrences for a similarity to be detected. This shows
the importance of the optimal threshold parameter. 

We now describe a dynamic programming algorithm for computing the
optimal threshold $T_F(m,k)$. For a binary word $w$, consider the
quantity $T_F(m,k,w)$ defined as the minimal number of occurrences of
seeds of $F$ in all $(m,k)$-si\-mi\-la\-ri\-ties which have the suffix $w$.
By definition, $T_F(m,k)=T_F(m,k,\varepsilon)$. Assume that we precomputed
values ${\cal T}_F(j,w)=T_F(s_{max},j,w)$, for all $j\leq
\max\{k,s_{max}\}$, $|w| = s_{max}$. 
The algorithm is based on the following recurrence relations on
$T_F(i,j,w)$, for $i\geq s_{max}$. 
\[ 
  T_F(i,j,w[1..n]) = 
    \left\{\!\!\begin{array}{ll}
      {\cal T}_F(j,w),      &  \textrm{if } i\!=\!s_{max}, \\
      T_F(i\!-\!1,j\!-\!1,w[1..n\!-\!1]),                                & \textrm{if } w[n]\!=\!0,\\
      T_F(i\!-\!1,j  ,w[1..n\!-\!1]) + [\sum_{l=1}^{L} suf\!f(Q_l,w)], & \textrm{if } n\!=\!s_{max},  \\
      \min\{T_F(i,j,\mathtt{1}.w),T_F(i,j,\mathtt{0}.w)\},&\textrm{if } zeros(w)\!<\!j,\\ 
      T_F(i,j,\mathtt{1}.w), & \textrm{if } zeros(w)\!=\!j.\\
    \end{array} \right.
\]

The first relation is an initial condition of the recurrence. 
The second one is based on the fact that if the last symbol of
$w$ is $0$, then no seed can match a suffix of $w$ (as the last
position of a seed is always assumed to be a matching position). 
The third relation reduces the size of the problem by counting the number of suffix seed occurrences. 
The fourth one splits the counting into two cases, by considering
two possible characters occurring on the left of $w$. If $w$ already
contains $j$ $0$'s, then only $1$ can occur on the left of $w$, as
stated by the last relation. 

A dynamic programming implementation of the above recurrence allows to
compute \linebreak[4]$T_F(m,k,\varepsilon)$ in a bottom-up fashion, starting from
initial values ${\cal T}_F(j,w)$ and applying the above relations in
the order in which they are given. 
A straightforward dynamic programming implementation requires 
$O(m\cdot k\cdot 2^{(s_{max}+1)})$ time and space. 
However, the space complexity can be immediately improved: if values of $i$ are processed successively, then only 
$O(k\cdot 2^{(s_{max}+1)})$ space is needed. 
Furthermore, for each $i$ and $j$, 
it is not necessary to consider all $2^{(s_{max}+1)}$
different strings $w$, but only those which contain up to $j$ $0$'s. 
The number of those $w$ is 
$g(j,s_{max}) =  \sum_{e=0}^{j}  {s_{max}\choose e}$. For each $i$,
$j$ ranges from $0$ to $k$. 
Therefore, for each $i$, we need to store 
$f(k,s_{max})=\sum_{j=0}^{k}g(j,s_{max})=\sum_{j=0}^{k}{s_{max}
  \choose j}\cdot(k-j+1)$ values. This yields the same space
complexity as 
for computing the optimal threshold for one seed
\cite{burkhardt03better}. 

The quantity $\sum_{l=1}^{L} suf\!f(Q_l,w)$ can be precomputed
for all considered words $w$ in time $O(L\cdot g(k,s_{max}))$ and space 
$O(g(k,s_{max}))$, under the assumption that checking an individual match
is done in constant time. 
This leads to the overall time complexity $O(m\cdot f(k,s_{max}) +
L\cdot g(k,s_{max}))$ with the leading term  $ m\cdot f(k,s_{max})$
(as $L$ is usually small compared to $m$ and $g(k,s_{max})$ is 
smaller than $f(k,s_{max})$).
 
\end{subsection}
\begin{subsection}{Number of undetected similarities}
\label{undetected}

We now describe a dynamic programming algorithm that computes another
characteristic of a seed family, that will be used later in
Section~\ref{heuristics}. Consider an $(m,k)$-problem. Given a seed family
$F=<Q_l>_{l=1}^{L}$, we are 
interested in the number $U_F(m,k)$ of $(m,k)$-similarities that are not detected by $F$. 
For a binary word $w$, define $U_F(m,k,w)$ to be the number of undetected
$(m,k)$-similarities that have the suffix $w$. 

Similar to \cite{KLMT04}, let $X(F)$ be the set of binary words $w$
such that {\em (i)} $|w|\leq s_{max}$,
{\em (ii)} for any $Q_l\in F$, $\suff(Q_l,1^{s_{max}-|w|}w)=0$, and
{\em (iii)} no
proper suffix of $w$ satisfies {\em (ii)}. Note that word $0$ belongs to
$X(F)$, as the last position of every seed is a matching position. 

The following recurrence relations allow to compute  $U_F(i,j,w)$ for
$i\leq m$, $j\leq k$, and $|w|\leq s_{max}$. 
\[ 
  U_F(i,j,w[1..n])=
      {
        \left\{\!\!\!\begin{array}{ll}
          {i - |w| \choose j - zeros(w)}, &  \textrm{if } i < s_{min},\\
          0,  & \textrm{if } \exists l\in [1..L], \suff(Q_l,w) = 1,\\
          U_F(i-1,j-last(w),w[1..n-1]), & \textrm{if } w \in X(F),\\
          U_F(i,j,\mathtt{1}.w) + U_F(i,j,\mathtt{0}.w), &\textrm{if } zeros(w)<j,\\
          U_F(i,j,\mathtt{1}.w), & \textrm{if } zeros(w)=j.\\
        \end{array} \right.
      }
\]
The first condition says that if $i < s_{min}$, then no word of length
$i$ will be detected, hence the binomial coefficient. The second condition
is straightforward. The third relation follows from the definition of
$X(F)$ and allows to reduce the size of the problem. The last two
conditions are similar to those from the previous section. 
  
The set $X(F)$ can be precomputed in time $O(L\cdot g(k,s_{max}))$ and
the worst-case time complexity of the whole algorithm remains 
$O(m\cdot f(k,s_{max})+L\cdot g(k,s_{max}))$.

\end{subsection}  

\begin{subsection}{Contribution of a seed} 
\label{contribution}
Using a similar dynamic programming technique, one can 
compute, for a given seed of the
family, the number of $(m,k)$-similarities that are detected only by
this seed and not by the others. Together with the number of
undetected similarities, this parameter will be used later in
Section~\ref{heuristics}. 

Given an $(m,k)$-problem and a family $F=<Q_l>_{l=1}^{L}$, we define
$S_F(m,k,l)$ to be the number of $(m,k)$-similarities detected by the
seed $Q_l$ exclusively (through one or several occurrences), and
$S_F(m,k,l,w)$ to be the number of those similarities ending with the
suffix $w$.
A dynamic programming algorithm similar to the one described in the
previous sections can be applied to compute $S_F(m,k,l)$. 
The recurrence is given below.

\[
   S_F(i,j,l,w[1..n]) =
     \left\{\!\!\!\begin{array}{ll}
     0                                       & \mbox{if~} i\!<\!\!s_{min}\!\mbox{~or~}\!\exists l'\!\!\neq\!l\ \suff(Q_{l'},\!w)\!=\!1 \\
     S_F(i\!-\!1,j\!-\!1 , l , w[1..n\!-\!1])  & \mbox{if~} w[n] = 0 \\
     S_F(i\!-\!1,j,l,w[1..n\!-\!1])      & \mbox{if~} n=|Q_l|   \mbox{~and~} \suff(Q_{l},w)=0 \\
     S_F(i\!-\!1,j,l,w[1..n\!-\!1])       & \mbox{if~} n=s_{max} \mbox{~and~} \suff(Q_{l},w)=1 \\
     ~+~U_F(i\!-\!1,j,w[1..n\!-\!1])     & \mbox{and~} \forall l'\neq l,\suff(Q_{l'},w)=0,\\
     S_F(i,j,l,\mathtt{1}.w[1..n]) & \\
     ~+~S_F(i,j,l,\mathtt{0}.w[1..n])& \mbox{if~} zeros(w) < j\\
     S_F(i,j,l,\mathtt{1}.w[1..n])& \mbox{if~} zeros(w)=j\\
    \end{array}\right.
\]

The third and fourth relations play the principal role: if $Q_l$ does
not match a suffix of $w[1..n]$, then we simply drop out the
last letter. If $Q_l$ matches a suffix of $w[1..n]$, but no other
seed does, then we count prefixes matched by $Q_l$ exclusively (term
$S_F(i-1,j,l,w[1..n-1])$) 
together with prefixes matched by no seed at all (term
$U_F(i-1,j,w[1..n-1])$). The latter is computed by the algorithm of
the previous section. 

The complexity of computing $S_F(m,k,l)$ for a given $l$ is the same as
the complexity of dynamic programming algorithms from the previous
sections.
\end{subsection}
\end{section}

\begin{section}{Seed design}
\label{seed-design}

In the previous Section we showed how to compute various useful
characteristics of a given family of seeds. A much more difficult task
is to find an efficient seed family that solves a given
$(m,k)$-problem.
Note that there exists a trivial
solution where the family consists of all ${m \choose k}$ position
combinations, but this is in general unacceptable in practice because
of a huge number of seeds. 
Our goal is to find families of reasonable size (typically, with the
number of seeds smaller than ten), with a good filtration efficiency. 

In this section, we present several results that contribute to this
goal. 
In Section~\ref{one-joker},
we start with the case of single seed with a fixed number of jokers
and show, in particular, that for one joker, there exists one best
seed in a sense that will be defined. We then show in
Section~\ref{expansion} that a solution for a 
larger problem can be obtained from a smaller one by a regular
expansion operation. In Section~\ref{regular}, we focus on seeds that
have a periodic structure and show how those seeds can be constructed
by iterating some smaller seeds. We then show a way to build efficient
families of periodic seeds.
Finally, in Section~\ref{heuristics}, we briefly describe a
heuristic approach to constructing efficient seed families that we
used in the experimental part of this work presented in
Section~\ref{experiments}.

\begin{subsection}{Single seeds with a fixed number of jokers} 
\label{one-joker}

Assume that we fixed a class of seeds under interest (e.g. seeds of a
given minimal weight). 
One possible way to define the seed design problem is to fix a
similarity length $m$ and find a seed that solves the
$(m,k)$-problem with the largest possible value of $k$. A
complementary definition is to fix $k$ and minimize $m$ provided that
the $(m,k)$-problem is still solved. In this section, we 
adopt the second definition and present an optimal solution for one
particular case. 

For a seed $Q$ and a number of mismatches $k$, define the {\em $k$-critical
length} for $Q$ as the minimal value $m$ such that $Q$ solves the
$(m,k)$-problem. For a class of seeds ${\cal C}$ and a value $k$, a
seed is $k$-optimal in ${\cal C}$ if $Q$ has the minimal $k$-critical
length among all seeds of ${\cal C}$. 

One interesting class of seeds ${\cal C}$ is obtained by putting an
upper bound on the possible number of jokers in the seed, i.e. on the
number $(s(Q)-w(Q))$. We have found a general solution of the seed
design problem for the class ${\cal C}_1(n)$ consisting of seeds of
weight $d$ with only one joker, i.e. seeds $\SHAPE{\#}^{d-r}\SHAPE{-}\SHAPE{\#}^{r}$.

Consider first the case of one mismatch, i.e. $k=1$. A $1$-optimal seed from ${\cal C}_1(d)$ is
$\SHAPE{\#}^{d-r}\SHAPE{-}\SHAPE{\#}^{r}$ with 
$r = \lfloor d/2 \rfloor$. To see this, consider an arbitrary seed
$Q=\SHAPE{\#}^{p}\SHAPE{-}\SHAPE{\#}^{q}$, $p+q=d$, and
assume by symmetry that 
$p\geq q$. Observe that the longest $(m,1)$-similarity that is
not detected by $Q$ is $1^{p-1}01^{p+q}$ of length $(2p+q)$. Therefore,
we have to minimize $2p+q=d+p$, and since $p\geq \lceil d/2 \rceil$,
the minimum is reached for $p= \lceil d/2 \rceil$, $q = \lfloor d/2 \rfloor$. 

However, for $k\geq 2$, an optimal seed has an asymmetric structure
described by the following theorem. 

\begin{theorem}
\label{best-seed}
Let $n$ be an integer and $r = [ d/3 ]$ ($[x]$ is the closest
integer to $x$). 
For every $k \geq 2$, seed $Q(d) = \SHAPE{\#}^{d-r}\SHAPE{-}\SHAPE{\#}^{r}$ 
is $k$-optimal among the seeds of ${\cal C}_1(d)$.
\end{theorem}
\begin{proof}
Again, consider a seed $Q=\SHAPE{\#}^{p}\SHAPE{-}\SHAPE{\#}^{q}$,
$p+q=d$, and assume that $p\geq q$. 
Consider the longest word $S(k)$ from $(1^* 0)^k 1^*$, $k\geq 1$, which
is not detected by $Q$ and let $L(k)$ is the length of $S(k)$. By the above remark,
$S(1)=1^{p-1}01^{p+q}$ and $L(1)=2p+q$. 

It is
easily seen that for every $k$, $S(k)$ starts either with 
$1^{p-1}0$, or with $1^{p+q}01^{q-1}0$. 
Define $L'(k)$ to be the maximal
length of a word from $(1^* 0)^k 1^*$ that is not detected by $Q$ and
starts with $1^{q-1}0$. Since prefix $1^{q-1}0$ implies no additional
constraint on the rest of the word, we have $L'(k)=q+L(k-1)$. 
Observe that 
$L'(1)=p+2q$ (word $1^{q-1}01^{p+q}$). To summarize, we have the
following recurrences for $k\geq 2$:
\begin{eqnarray}
L'(k) & = & q+L(k-1),\\
L(k) & = & \max\{p+L(k-1),p+q+1+L'(k-1)\},\label{L}
\end{eqnarray}
with initial conditions $L'(1)=p+2q$, $L(1)=2p+q$.

Two cases should be distinguished. If $p\geq 2q+1$, then the
straightforward induction shows that the first term in (\ref{L}) is
always greater, and we have 
\begin{equation}
L(k)=(k+1)p+q, \label{p>2q+1}
\end{equation}
and the corresponding longest word is 
\begin{equation}
S(k)=( 1^{p-1} 0 )^k 1^{p+q}.
\end{equation}

If $q\leq p\leq 2q+1$, then by induction, we obtain
\begin{eqnarray}
L(k)=\begin{cases}
(\ell+1)p+(k+1)q+\ell &\mbox{ if~} k = 2\ell,\\
(\ell+2)p+kq+\ell &\mbox{ if~} k = 2\ell+1,
\end{cases} \label{p<2q+1}
\end{eqnarray}
and
\begin{eqnarray}
S(k)=  \begin{cases}
    ( 1^{p+q} 0 1^{q-1} 0 )^\ell 1^{p+q}          &\mbox{ if~} k = 2\ell,\\
     1^{p-1} 0 (1^{p+q} 0 1^{q-1} 0)^\ell 1^{p+q}&\mbox{ if~} k = 2\ell+1.
  \end{cases} 
\end{eqnarray}

By definition of $L(k)$, seed $\SHAPE{\#}^{p}\SHAPE{-}\SHAPE{\#}^{q}$
detects any word from $(1^* 0)^k 1^*$ of length $(L(k)+1)$ or more,
and this is the tight bound. Therefore, we have to find $p,q$ which
minimize $L(k)$. Recall that $p+q=d$, and observe that for $p\geq
2q+1$, $L(k)$ (defined by (\ref{p>2q+1})) 
is increasing on $p$, while for $p\leq 2q+1$, $L(k)$ (defined by
(\ref{p<2q+1})) is decreasing on $p$. Therefore, both functions reach
its minimum when $p=2q+1$. Therefore, if $d\equiv 1 \pmod{3}$, we
obtain $q=\lfloor d/3\rfloor$ and $p=d-q$. If $d\equiv 0 \pmod{3}$, a
routine computation shows that the minimum is reached at $q=d/3$,
$p=2d/3$, and if $d\equiv 2 \pmod{3}$, the minimum is reached at
$q=\lceil d/3\rceil$, $p=d-q$. Putting the three cases together
results in $q=[d/3]$, $p=d-q$. 
\end{proof}

To illustrate Theorem~\ref{best-seed}, seed \SHAPE{\#\#\#\#-\#\#} is optimal among
all seeds of weight $6$ with one joker. This means that this seed solves
the $(m,2)$-problem for all $m \geq 16$ and this is the smallest
possible bound over all seeds of this class. Similarly, 
this seed solves the $(m,3)$-problem for all $m \geq 20$, which is the
best possible bound, etc.
\end{subsection}

\begin{subsection}{Regular expansion and contraction of seeds}
\label{expansion}
We now show that seeds solving larger problems can be obtained from seeds
solving smaller problems, and vice versa, using regular expansion and
regular contraction operations. 

Given a seed $Q$ , its {\em $i$-regular expansion} $i \otimes Q $ is obtained
by multiplying each matching position by $i$. This is equivalent to
inserting $i-1$ jokers between every two successive positions along
the seed. For example, if $Q = \{0,2,3,5\}$ (or \SHAPE{\#-\#\#-\#}),
then the $2$-regular expansion of $Q$ is $2\otimes Q = \{0,4,6,10\}$
(or  \SHAPE{\#---\#-\#---\#}). Given a family $F$, its $i$-regular
expansion $i \otimes F$ is the family obtained by applying the
$i$-regular expansion on each seed of $F$. 

\begin{lemma} 
\label{expand}
If a family $F$ solves an $(m,k)$-problem,
then the $(i m,(i+1)k - 1)$-problem is solved both by family $F$ and by
its $i$-regular expansion $F_i = i \otimes F$. 
\end{lemma}
\begin{proof}
Consider an $(im,(i+1)k-1)$-similarity $w$. By the pigeon
hole principle, it contains at least one substring of length $m$ with
$k$ mismatches or less, and therefore $F$ solves the
$(im,(i+1)k-1)$-problem. On the other hand, consider $i$ disjoint
subsequences of $w$ each one consisting of $m$ positions
equal modulo $i$. Again, by the pigeon hole principle, at least one of
them contains $k$ mismatches or less, and therefore the 
$(im,(i+1)k-1)$-problem is solved by $i \otimes F$. 
\end{proof}

The following lemma is the inverse of Lemma~\ref{expand}, it states
that if seeds solving a bigger problem have a regular structure, then
a solution for a smaller problem can be obtained by the regular
contraction operation, inverse to the regular expansion. 

\begin{lemma} 
If a family $F_i = i \otimes F$ solves an $(im,k)$-problem, then
$F$ solves both the $(i m,k)$-problem and the $(m,\lfloor k/i
\rfloor)$-problem. 
\end{lemma}
\begin{proof}
One can even show that $F$ solves the $(i m,k)$-problem with the
additional restriction for $F$ to match inside one of
the position intervals $[1..m],[m+1..2m],\ldots,[(i-1)m+1..im]$. This is done by
using the bijective mapping from Lemma~\ref{expand}: given 
an $(im,k)$-similarity $w$, consider $i$ disjoint subsequences
$w_j$ ($0\leq j\leq i-1$) of $w$ obtained by picking $m$ positions
equal to $j$ modulo $i$, and then consider the concatenation
$w'=w_1w_2\ldots w_{i-1}w_0$. 

For every $(im,k)$-similarity $w'$, its inverse image $w$ is detected
by $F_i$, and therefore $F$ detects $w'$ at one of the intervals
$[1..m],[m+1..2m],\ldots,[(i-1)m+1..im]$. Futhermore, for any
$(m,\lfloor k/i\rfloor)$-similarity $v$, consider $w'=v^i$ and its
inverse image $w$. As $w'$ is detected by $F_i$, $v$ is
detected by $F$. 
\end{proof}

\begin{example}
To illustrate the two lemmas above, we give the following example
pointed out in \cite{burkhardt03better}. The following two seeds are
the only seeds of weight $12$ that solve the $(50,5)$-problem:
\SHAPE{\#-\#-\#---\#-----\#-\#-\#---\#-----\#-\#-\#---\#} and
\SHAPE{\#\#\#-\#--\#\#\#-\#--\#\#\#-\#}. The first one is the
$2$-regular expansion of the second. The second one is the only seed
of weight $12$ that solves the $(25,2)$-problem.
\end{example}

The regular expansion allows, in some cases, to obtain an efficient
solution for a larger problem by reducing it to a smaller problem for
which an optimal or a near-optimal solution is known.
\end{subsection}
  
\begin{subsection}{Periodic seeds}   
\label{regular}
    
In this section, we study seeds with a periodic structure that can be
obtained by iterating a smaller seed. Such seeds often turn out to be
among maximally weighted seeds solving a given
$(m,k)$-problem. Interestingly, this contrasts with the lossy
framework where optimal seeds usually have a ``random'' irregular
structure. 
   
Consider two seeds $Q_1,\!Q_2$ represented as words over
$\{\SHAPE{\#},\!\SHAPE{-}\}$. In this section, we lift the assumption
that a seed must start and end with a matching position. We denote \mbox{$[Q_1,\!Q_2]^i$} the seed defined as
$(Q_1 Q_2)^{i} Q_1$. For example, \mbox{$[\SHAPE{\#\#\#-\#},\!\SHAPE{--}]^2\!=\!\SHAPE{\#\#\#-\#--\#\#\#-\#--\#\#\#-\#}$}.

We also need a modification of the $(m,k)$-problem, where
$(m,k)$-similarities are considered modulo a cyclic permutation. We
say that a seed family $F$ solves a {\em cyclic $(m,k)$-problem}, if
for every $(m,k)$-similarity $w$, $F$ detects one of cyclic
permutations of $w$. Trivially, if $F$
solves an $(m,k)$-problem, it also solves the cyclic
$(m,k)$-problem. To distinguish from a cyclic problem, we call
sometimes an $(m,k)$-problem a {\em linear} problem. 

We first restrict ourselves to the single-seed case. The following
lemma demonstrates that iterating smaller seeds solving a cyclic
problem allows to obtain a solution for bigger problems, for the same
number of mismatches. 

\begin{lemma}
\label{cyclic}
If a seed $Q$ solves a {\em cyclic}
$(m,k)$-problem, then for every $i\geq 0$, the seed 
$Q_i =[Q,-^{(m-s(Q))}]^i$ 
solves the linear 
$(m \cdot (i+1)+s(Q)- 1,k)$-problem. If $i\neq 0$, the inverse holds
too. 
\end{lemma}
\begin{proof}
$\Rightarrow$ 
Consider an $(m \cdot (i+1) + s(Q)- 1,k)$-similarity
$u$. Transform $u$ into a similarity $u'$ for the cyclic
$(m,k)$-problem as follows. For each mismatch position $\ell$ of $u$,
set $0$ at position $(\ell\mod m)$ in $u'$. The other
positions of $u'$ are set to $1$. Clearly, there are at
most $k$ $0$'s in $u$. 
As $Q$ solves the $(m,k)$-cyclic problem, we can find at least one
position $j$, $1\leq j \leq m$, such that $Q$ detects $u'$ cyclicly. 
        
We show now that $Q_i$ matches at position $j$ of $u$ (which is a
valid position as $1 \leq j\leq m$ and $s(Q_i)=im + s(Q)$). As the
positions of $1$ in $u$ are projected modulo $m$ to
matching positions of $Q$, then there is no $0$ under any matching
element of $Q_i$, and thus $Q_i$ detects $u$.
        
$\Leftarrow$ 
Consider a seed $Q_i= [Q,-^{(m-s(Q))}]^i$ solving the 
$(m\cdot (i+1) + s(Q) - 1,k)$-problem. As $i>0$, consider
$(m\cdot (i+1) + s(Q) - 1,k)$-similarities having all their mismatches
located inside the interval $[m,2m-1]$. For each such similarity,
there exists a position $j$, $1\leq j \leq m$, such that $Q_i$ detects
it. 
Note that
the span of $Q_i$ is at least $m + s(Q)$, which implies that there is
either an entire occurrence of $Q$ inside the window
$[m,2m-1]$, or a prefix of $Q$ matching a suffix of the window and the
complementary suffix of $Q$ matching a prefix of the window. 
This implies that $Q$ solves the cyclic $(m,k)$-problem.
\end{proof}

\begin{example}
\label{example1}
Observe that the seed \SHAPE{\#\#\#-\#} solves the cyclic $(7,2)$-problem.
From Lemma~\ref{cyclic}, this implies that 
for every $i\geq 0$, the $(11+7 i,2)$-problem is solved by the seed
$[\SHAPE{\#\#\#-\#},\SHAPE{--}]^i$ of span $5 + 7i$. 
Moreover, for $i=1,2,3$, this seed is optimal (maximally weighted) over all seeds
solving the problem. 

By a similar argument based on Lemma~\ref{cyclic}, the periodic seed
$[\SHAPE{\#\#\#\#\#-\#\#},\SHAPE{---}]^i$ solves the 
$(18 + 11 i,2)$-problem. Note that its weight grows as $\frac{7}{11}m$
compared to $\frac{4}{7} m$ for the seed from the 
previous paragraph. However, when $m\rightarrow\infty$, this is not an asymptotically optimal
bound, as we will see later. 

The $(18 + 11 i,3)$-problem is solved by the seed
$(\SHAPE{\#\#\#-\#--\#},\SHAPE{---})^i$, as seed \SHAPE{\#\#\#-\#--\#}
solves the cyclic $(11,3)$-problem. For $i=1,2$, the former is a
maximally weighted seed among all solving the $(18+11i,3)$-problem. 
\end{example}

One question raised by these examples is whether iterating some seed
could provide an asymptotically optimal solution, i.e. a seed of 
maximal
asymptotic weight. The following theorem establishes a tight asymptotic
bound on the weight of an optimal seed, for a fixed number of
mismatches. It gives a negative answer to this question, as it shows
that the maximal weight grows faster than any linear fraction of the
similarity size. 

\begin{theorem}
\label{asympt}
Consider a constant $k$. Let $w(m)$ be the maximal weight of a seed
solving the cyclic $(m,k)$-problem. Then $(m-w(m))=\Theta(m^{\frac{k-1}{k}})$.
\end{theorem}
\begin{proof}
Note first that all seeds solving cyclic a $(m,k)$-problem can be considered
as seeds of span $m$. The number of jokers in any seed $Q$ is then
$n = m-w(Q)$. The theorem states that the minimal number
of jokers of a seed solving the $(m,k)$-problem is
$\Theta(m^{\frac{k-1}{k}})$ for every fixed $k$.

{\it Lower bound}
Consider a cyclic $(m,k)$-problem. 
The number $D(m,k)$ of distinct cyclic $(m,k)$-similarities satisfies
\begin{equation}
\frac{{m \choose k}}{m} \leq D(m,k),
\label{1}
\end{equation}
as every linear $(m,k)$-similarity has at most $m$ cyclicly equivalent
ones. 
Consider a seed $Q$. Let 
$n$ be the number of jokers in $Q$ and $J_Q(m,k)$ the
number of distinct cyclic $(m,k)$-similarities 
detected by $Q$. Observe that $J_Q(m,k) \leq {n \choose k}$
and if $Q$ solves the cyclic $(m,k)$-problem, then 
\begin{equation}
D(m,k) = J_Q(m,k) \leq {n \choose k}
\label{2}
\end{equation}
From (\ref{1}) and (\ref{2}), we have
\begin{equation}
\frac{{m \choose k}}{m} \leq {n \choose  k}.
\label{3}
\end{equation}
Using the Stirling formula, this gives $n(k) = \Omega(m^{\frac{k-1}{k}})$.

\begin{figure}[htb]
  \begin{center}
    \includegraphics[width=8cm]{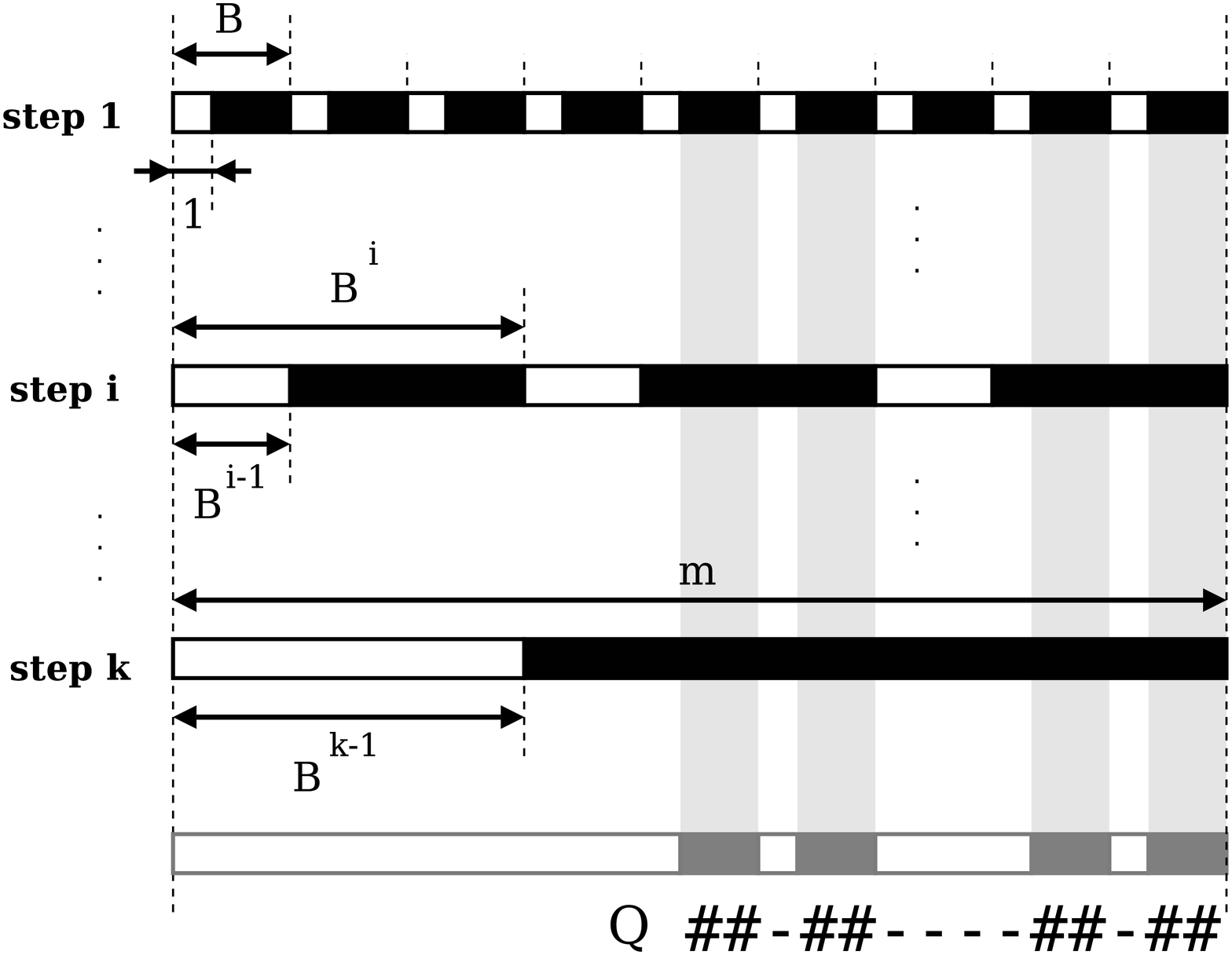}
    \caption{\label{seedbuilt}Construction of seeds $Q_i$ from the
      proof of Theorem~\ref{asympt}. 
Jokers are represented in white and matching positions in black. 
      }
  \end{center}
\end{figure}

{\it Upper bound}
To prove the upper bound, we construct a seed $Q$ that has no more
then $ k \cdot m^{\frac{k-1}{k}}$ joker positions and solves the cyclic
$(m, k)$-problem.

We start with the seed
$Q_0$ of span $m$ with all matching positions, and introduce jokers
into it in $k$ steps. After step $i$, the obtained seed is denoted
$Q_i$, and $Q=Q_k$. 

Let $B = \lceil m^{\frac{1}{k}} \rceil$. 
$Q_1$ is obtained by introducing into $Q_0$ individual jokers with
periodicity $B$ by placing jokers at positions $1,B+1,2B+1,\ldots$. 
At step 2, we introduce into $Q_1$ contiguous intervals of jokers of
length $B$ with periodicity $B^2$, such that 
jokers are placed at positions $[1\ldots B], [B^2+1\ldots B^2+B],[2B^2+1 \ldots 2B^2+B],\ldots$.

In general, at step $i$ ($i \leq k$), we introduce into $Q_{i}$ intervals of $B^{i-1}$ jokers with periodicity $B^i$ at
positions $[1 \ldots B^{i-1}], [ B^i+1 \ldots B^i+B^{i-1}],\ldots$ 
(see Figure~\ref{seedbuilt}).

Note that $Q_i$ is periodic with periodicity $B^{i}$. Note also that
at each step $i$, we introduce at most $\lfloor m^{1-\frac{i}{k}}
\rfloor$ intervals of $B^{i-1}$ jokers. Moreover, due to overlaps with
already added jokers, each interval adds $(B-1)^{i-1}$ new jokers. 

This implies that the total number of jokers added at step $i$ is at most
$m^{1-\frac{i}{k}} \cdot (B-1)^{i-1} \leq m^{1-\frac{i}{k}} \cdot
m^{\frac{1}{k} \cdot (i-1)} = m^{\frac{k-1}{k}}$.
Thus, the total number of jokers in $Q$ is less than $k \cdot m^{\frac{k-1}{k}}$.

By induction on $i$, we prove that for any $(m,i)$-similarity $u$
($i\leq k$), $Q_i$ detects $u$ cyclicly, that is there is a cyclic shift
of $Q_i$ 
such that all $i$ mismatches of $u$
are covered with jokers introduced at steps $1,\ldots,i$. 

For $i=1$, the statement is obvious, as we can always cover the single
mismatch by shifting $Q_1$ by at most $(B-1)$ positions. 
Assuming that the statement holds for $(i-1)$, we show now that it holds for $i$ too. 
Consider an $(m,i)$-similarity $u$. Select one mismatch of $u$. By
induction hypothesis, the other $(i-1)$ mismatches can be covered by
$Q_{i-1}$. Since $Q_{i-1}$ has period $B^{i-1}$ and $Q_i$ differs from
$Q_{i-1}$ by having at least one contiguous interval of $B^{i-1}$
jokers, we can always shift $Q_i$ by $j\cdot B^{i-1}$ positions such
that the selected mismatch falls into this interval. This shows that
$Q_i$ detects $u$. We conclude that $Q$ solves the cyclic
$(m,i)$-problem.
\end{proof}

Using Theorem~\ref{asympt}, we obtain the following bound on the
number of jokers for the {\em linear} $(m,k)$-problem. 
\begin{lemma}
\label{asympt-linear}
Consider a constant $k$. Let $w(m)$ be the maximal weight of a seed
solving the linear $(m,k)$-problem. Then $(m-w(m))=\Theta(m^{\frac{k}{k+1}})$.
\end{lemma}
\begin{proof}
To prove the upper bound, we construct a seed $Q$ that solves the
linear $(m,k)$-problem and satisfies the asymptotic bound. 
Consider some $l<m$ that will be defined later, and let $P$ be a
seed that solves the cyclic $(l,k)$-problem. 
Without loss of generality, we assume $s(P)=l$. 

For a real number $e \geq 1$,
define $P^e$ to be the maximally weighted seed of span at most $l^e$
of the form $ P' \cdot P\cdots P \cdot  P''$, where $P'$ and $P''$
are respectively a suffix and a prefix of $P$. 
Due to the condition of maximal weight, $w(P^e)\geq e\cdot w(P)$. 

We now set $Q=P^e$ for some real $e$ to be defined. Observe that 
if $e\cdot l\leq m-l$, then $Q$ solves the linear
$(m,k)$-problem. Therefore, we set $e=\frac{m-l}{l}$. 

From the proof of Theorem~\ref{asympt}, we have $l-w(P)\leq k\cdot 
l^\frac{k-1}{k}$. We then have
\begin{equation}{\label{equation:linear01}}
  w(Q) = e \cdot w(P) \geq \frac{m-l}{l} \cdot (l - k \cdot l^{\frac{k-1}{k}}).
\end{equation}
If we set
\begin{equation}{\label{equation:linear02}}
  l = m^{\frac{k}{k+1}},
\end{equation}
we obtain 
\begin{equation}{\label{equation:linear03}}
  m - w(Q)  \leq (k+1) {m}^{\frac{k}{k+1}} -k
  {m}^{\frac{k-1}{k+1}},
\end{equation}
and as $k$ is constant, 
\begin{equation}{\label{equation:linear04}}
  m - w(Q) = \mathcal{O}(m^{\frac{k}{k+1}}).
\end{equation}

The lower bound is obtained similarly to Theorem~\ref{asympt}. Let $Q$
be a seed solving a linear $(m,k)$-problem, and let $n=m-w(Q)$. From
simple combinatorial considerations, we have 
\begin{equation}
{m \choose k} \leq {n \choose  k}\cdot (m-s(Q))\leq {n \choose
  k}\cdot n,
\end{equation}
which implies $n=\Omega(m^\frac{k}{k+1})$ for constant $k$.
\end{proof}

The following simple lemma is also useful for constructing efficient
seeds. 
\begin{lemma}
\label{cut}
Assume that a family  $F$ solves an $(m,k)$-problem. Let $F'$ be the family
obtained from $F$ by cutting out $l$ characters from the left and $r$
characters from the right of each seed of $F$. Then $F'$ solves the
$(m-r-l,k)$-problem. 
\end{lemma}

\begin{example}
The $(9+7i,2)$-problem is solved by the seed
$[\SHAPE{\#\#\#},\SHAPE{-\#--}]^i$ which is optimal for $i=1,2,3$. 
Using Lemma~\ref{cut}, this seed can be immediately obtained from the seed
$[\SHAPE{\#\#\#-\#},\SHAPE{--}]^i$
from Example~\ref{example1}, solving the $(11+7i,2)$-problem.
\end{example}

We now apply the above results for the single seed case to the case of
multiple seeds. 

For a seed $Q$ considered as a word over \{\SHAPE{\#},\SHAPE{-}\}, we
denote by $Q_{[i]}$ its cyclic shift to the left by $i$
characters. For example, if $Q = \SHAPE{\#\#\#\#-\#-\#\#--}$, then
$Q_{[5]} = \SHAPE{\#-\#\#--\#\#\#\#-}$. 
The following lemma gives a way to construct seed families solving
bigger problems from an individual seed solving a smaller cyclic
problem.

\begin{lemma}
\label{lemma-final}
Assume that a seed $Q$ solves a cyclic $(m,k)$-problem and assume that
$s(Q)=m$ (otherwise we pad $Q$ on the right with $(m-s(Q))$ jokers). 
Fix some $i>1$. 
For some $L>0$, consider a list of $L$ integers 
$0 \leq j_1 < \cdots < j_L < m$, and define a family of seeds
$F=<\| (Q_{[j_l]})^i\|>_{l=1}^L$, where $\| (Q_{[j_l]})^i\|$ stands for the
seed obtained from $(Q_{[j_l]})^i$ by deleting the joker characters at
the left and right edges. 
Define $\delta(l)=((j_{l-1}-j_{l})\mod m)$ (or, alternatively,
$\delta(l)=((j_{l}-j_{l-1})\mod m)$) for all $l$, $1\leq l\leq L$. 
Let $m'=\max\{s(\| (Q_{[j_l]})^i \|)+\delta(l)\}_{l=1}^L -1$.
Then $F$ solves the 
$(m',k)$-problem.
\end{lemma}

\begin{proof}
The proof is an extension of the proof of Lemma \ref{cyclic}.
Here, the seeds of the family are constructed in such a way that for
any instance of the linear $(m',k)$-problem, there exists at least one
seed that satisfies the property required in the proof of Lemma
\ref{cyclic} and therefore matches this instance. 
\end{proof}

In applying Lemma~\ref{lemma-final}, integers $j_l$ are chosen from the
interval $[0,m]$ in such a way that values $s(||(Q[j_l])^i||) + \delta(l)$
are closed to each other. 
We illustrate Lemma~\ref{lemma-final} with two
examples that follow.

\begin{example}
Let $m=11$, $k=2$. Consider the seed 
$Q = \SHAPE{\#\#\#\#-\#-\#\#--}$ solving the cyclic
$(11,2)$-problem. Choose $i=2$, $L=2$, $j_1 = 0$, $j_2 = 5$. 
This gives two seeds 
$Q_1=\|(Q_{[0]})^2\|=\SHAPE{\#\#\#\#-\#-\#\#--\#\#\#\#-\#-\#\#}$ and 
$Q_2=\|(Q_{[5]})^2\|=\SHAPE{\#-\#\#--\#\#\#\#-\#-\#\#--\#\#\#\#}$
of span $20$ and $21$ respectively, $\delta(1)=6$ and $\delta(2)=5$. 
$\max\{20+6,21+5\}-1=25$. Therefore, family $F=\{Q_1,Q_2\}$ solves the
$(25,2)$-problem.
\end{example}

\begin{example}
Let $m=11$, $k=3$. The seed $Q = \SHAPE{\#\#\#-\#--\#---}$ solving the
cyclic $(11,3)$-problem. Choose $i=2$, $L=2$, $j_1 = 0$, $j_2 =
4$. The two seeds are
$Q_1=\|(Q_{[0]})^2\|=\SHAPE{\#\#\#-\#--\#---\#\#\#-\#--\#}$ (span
$19$) and
$Q_2=\|(Q_{[4]})^2\|=\SHAPE{\#--\#---\#\#\#-\#--\#---\#\#\#}$ (span
$21$), with  $\delta(1)=7$ and
$\delta(2)=4$. $\max\{19+7,21+4\}-1=25$. Therefore, family
$F=\{Q_1,Q_2\}$ solves the $(25,3)$-problem.
\end{example}
\end{subsection}

\begin{subsection}{Heuristic seed design}
\label{heuristics}
Results of Sections~\ref{one-joker}-\ref{regular} allow one to construct
efficient seed families in certain cases, but still do not allow 
a systematic seed design. 
Recently, linear programming approaches to designing efficient seed families were
proposed in \cite{XuBrownLiMaCPM04} and in \cite{Brown2004},
respectively for DNA
and protein similarity search. However, neither of these methods 
aims at constructing lossless families. 

In this section, we outline
a heuristic genetic programming algorithm for designing lossless seed
families. The algorithm will be used in the experimental part of this
work, that we present in the next section. Note that this algorithm
uses dynamic programming algorithms of Section~\ref{dyn-progr}. 
Since the algorithm uses standard genetic programming techniques, we
give only a high-level description here without going into all details. 

The algorithm tries to iteratively improve characteristics of a
{\em population} of seed families until it finds a small family that
detects all $(m,k)$-similarities (i.e. is lossless). 
The first step of each iteration is based on screening current
families against a set of {\em difficult similarities} that are
similarities that have been detected by fewer 
families. This set is continually reordered and updated according
to the number of families that don't detect those similarities. 
For this, each set is stored in a tree and the reordering is done
using the {\em list-as-a-tree} principle \cite{OommenDong97}: each
time a similarity is not detected by a family, it is moved towards the
root of the tree such that its height is divided by two. 

For those families that pass through the screening, the number
of undetected similarities is computed by the dynamic programming
algorithm of Section~\ref{undetected}. The family is kept if it produces
a smaller number than the families currently known. An undetected
similarity obtained during this computation is added as a leaf to the
tree of difficult similarities. 

To detect seeds to be improved inside a family, we compute the
contribution of each seed by the dynamic programming algorithm of
Section~\ref{contribution}. The seeds with the least contribution are
then modified with a higher probability. In general, the population of
seed families is evolving by 
{\em mutating} and {\em crossing over} according to the set of
similarities they do not detect. Moreover, random seed families are
regularly injected into the population in order to avoid local
optima.

The described heuristic procedure 
often allows efficient or even optimal
solutions to be computed in a reasonable time. For example, in ten runs of the
algorithm, we found 3 of the 6 existing families of two seeds of weight
14 solving the $(25,2)$-problem. The whole computation took less than 1
hour, compared to a week of computation needed to exhaustively test
all seed pairs. 
Note that the randomized-greedy approach (incremental completion of
the seed set by adding the best random seed) applied a dozen of times
to the same problem yielded only sets of three and sometimes four, but
never two seeds, taking about 1 hour at each run. 

\end{subsection}
\end{section}

\begin{section}{Experiments}\label{experiments}

We describe two groups of experiments that we made. The first one
concerns the design of efficient seed families, and the second one
applies a multi-seed lossless filtration to the identification of
unique oligos in a large set of EST sequences. 

\begin{subsection}*{Seed design experiments}
\begin{table*}[htb]
  \begin{center}
   \caption{\label{seeds25_2} Seed families for $(25,2)$-problem}
    {\scriptsize
      \textsf{
        \begin{tabular}{|llll|}
          \hline
          size & weight & family seeds & $\delta$\\    
          \hline
          1   & 12${}^{e,p,g}$ & {\scriptsize
              \SHAPE{\#\#\#-\#--\#\#\#-\#--\#\#\#-\#}  
}             & $5.96\cdot 10^{-8}$\\
            2   & 14${}^{e,p,g}$ &
            {\scriptsize\SHAPE{\#\#\#\#-\#-\#\#--\#\#\#\#-\#-\#\#}}
& $7.47\cdot 10^{-9}$\\
                &            & {\scriptsize\SHAPE{~~~~~\#-\#\#--\#\#\#\#-\#-\#\#--\#\#\#\#}}                                             &                \\
            3   & 15${}^{p}$ & {\scriptsize\SHAPE{\#--\#\#-\#-\#\#\#\#\#\#--\#\#-\#-\#\#}}                                               & $2.80\cdot 10^{-9}$\\  
                &            & {\scriptsize\SHAPE{~~~~~~\#-\#\#\#\#\#\#--\#\#-\#-\#\#\#\#\#}}                                            &                \\
                &            & {\scriptsize\SHAPE{~~~~~~~~~~\#\#\#\#--\#\#-\#-\#\#\#\#\#\#--\#\#}}                                       &                \\
            4   & 16${}^{p}$ & {\scriptsize\SHAPE{\#\#\#-\#\#-\#-\#\#\#--\#\#\#\#\#\#\#}}                                                & $9.42\cdot 10^{-10}$\\
                &            & {\scriptsize\SHAPE{~~~~\#\#-\#-\#\#\#--\#\#\#\#\#\#\#-\#\#-\#}}                                           &                 \\
                &            & {\scriptsize\SHAPE{~~~~~~~~~\#\#\#--\#\#\#\#\#\#\#-\#\#-\#-\#\#\#}}                                       &                 \\ 
                &            & {\scriptsize\SHAPE{~~~~~~~~~~~~~~\#\#\#\#\#\#\#-\#\#-\#-\#\#\#--\#\#\#}}                                  &                 \\
           6    & 17${}^{p}$ & {\scriptsize\SHAPE{\#\#-\#-\#\#--\#\#\#\#\#\#\#-\#\#\#\#-\#}}                                             & $3.51\cdot 10^{-10}$\\
                &            & {\scriptsize\SHAPE{~~~\#-\#\#--\#\#\#\#\#\#\#-\#\#\#\#-\#-\#\#}}                                          &                 \\
                &            & {\scriptsize\SHAPE{~~~~~~~~~\#\#\#\#\#\#\#-\#\#\#\#-\#-\#\#--\#\#\#}}                                     &                 \\
                &            & {\scriptsize\SHAPE{~~~~~~~~~~~~~\#\#\#-\#\#\#\#-\#-\#\#--\#\#\#\#\#\#\#}}                                 &                 \\
                &            & {\scriptsize\SHAPE{~~~~~~~~~~~~~~~~~\#\#\#\#-\#-\#\#--\#\#\#\#\#\#\#-\#\#\#}}                             &                 \\
                &            & {\scriptsize\SHAPE{~~~~~~~~~~~~~~~~~~~~~~~~\#\#--\#\#\#\#\#\#\#-\#\#\#\#-\#-\#\#--\#}}                    &                 \\
         \hline%
       \end{tabular}%
       }%
     }
 \end{center}%
\end{table*}%
\begin{table*}[h!]
\begin{center}
    \caption{\label{seeds25_3} Seed families for $(25,3)$-problem}%
    {\scriptsize%
      \textsf{%
        \begin{tabular}{|llll|}%
          \hline%
          size & weight & family seeds & $\delta$\\ 
          \hline
          1 & 8 ${}^{e,p,g}$  & \SHAPE{\#\#\#-\#-----\#\#\#-\#}%
  & $1.53\cdot 10^{-5}$\\
          2 & 10${}^{p}$  & \SHAPE{\#\#\#\#-\#-\#\#--\#---\#\#}                                           & $1.91\cdot 10^{-6}$\\
            &             & \SHAPE{~~~~~~~\#\#--\#---\#\#\#\#-\#-\#\#}                                    & \\
          3 & 11${}^{p}$  & \SHAPE{\#---\#\#\#\#-\#-\#\#--\#---\#\#}                                      & $7.16\cdot 10^{-7}$\\
            &             & \SHAPE{~~~~~\#\#\#-\#-\#\#--\#---\#\#\#\#}                                    & \\
            &             & \SHAPE{~~~~~~~~~~~\#\#--\#---\#\#\#\#-\#-\#\#--\#}                            & \\
          4 & 12${}^{p}$  & \SHAPE{\#---\#\#\#\#-\#-\#\#--\#---\#\#\#}                                  & $2.39\cdot 10^{-7}$\\
            &             & \SHAPE{~~~~~\#\#\#-\#-\#\#--\#---\#\#\#\#-\#}                               & \\
            &             & \SHAPE{~~~~~~~~~\#-\#\#--\#---\#\#\#\#-\#-\#\#--\#}                         & \\
            &             & \SHAPE{~~~~~~~~~~~\#\#--\#---\#\#\#\#-\#-\#\#--\#---\#}                     & \\
          \hline    
        \end{tabular}%
        }}%
\end{center}
\end{table*}%

We considered several $(m,k)$-problems. For each problem, and for a
fixed number of seeds in the family, we computed families solving
the problem and realizing the largest possible seed weight (under a natural
assumption that all seeds in a family have the same weight). We also
kept track of the ways (periodic seeds, genetic
programming heuristics, exhaustive search) in which those families can
be computed. 

Tables~\ref{seeds25_2} and \ref{seeds25_3} summarize some results
obtained for the $(25,2)$-problem and the $(25,3)$-problem
respectively. Families of periodic seeds (that can be found using
Lemma~\ref{lemma-final}) are marked with ${}^{p}$, those that are
found using a genetic algorithm are marked with ${}^{g}$, and those
which are obtained by an exhaustive search are marked with ${}^{e}$. 
Only in this latter case, the families are guaranteed to be optimal. 
Families of periodic seeds are shifted according to their construction
(see Lemma~\ref{lemma-final}).

Moreover, to compare the selectivity of different families solving a given
$(m,k)$-problem, we estimated the probability $\delta$ for at least one of the
seeds of the family to match at a given position of a uniform
Bernoulli four-letter sequence. This has been done using
the inclusion-exclusion formula. 

Note that the simple fact of passing from a single seed to a two-seed
family results in a considerable gain in efficiency: in both examples
shown in the tables there a change of about one order magnitude in the
selectivity estimator $\delta$. 

\end{subsection}

\begin{subsection}*{Oligo selection using multi-seed filtering}

An important practical application of lossless filtration is the
selection of reliable oligonucleotides for DNA micro-array
experiments. Oligonucleotides (oligos) are small {DNA} sequences of
fixed size (usually ranging from $10$ to $50$) designed to hybridize
only with a specific region of the genome sequence. In micro-array
experiments, oligos are expected to match ESTs that stem from a given
gene and not to match those of other genes. As the first
approximation, the problem of oligo
selection can then be formulated as the search for strings of a fixed
length that occur in a given sequence but do not occur, within a
specified distance, in other sequences of a given (possibly very
large) sample. Different approaches to this problem apply different
distance measures and different algorithmic techniques
\cite{Stormo01,Kaderali02,Rahmann03,Zheng03}.
The experiments we briefly present here demonstrate that the
multi-seed filtering provides an efficient computation of candidate
oligonucleotides. 
These should then be further processed by complementary methods in order to
take into account other physico-chemical factors occurring in
hybridisation, such as the melting temperature or the possible hairpin structure
of palindromic oligos.

Here we adopt the formalization of the oligo selection problem as the problem of identifying
  in a given sequence (or a sequence database) all substrings of length $m$
  that have no occurrences elsewhere in the sequence within the Hamming
  distance $k$. The parameters $m$ and $k$ were set to $32$ and
  $5$ respectively. For the $(32,5)$-problem, different seed families
  were designed and their selectivity was estimated. Those are summarized
  in the table in  Figure~\ref{seeds32_5}, using the same conventions as
  in Tables~\ref{seeds25_2} and \ref{seeds25_3} above. The family
  composed of 6 seeds of weight 11 was selected for the filtration
  experiment (shown in Figure~\ref{seeds32_5}). 

  \begin{figure}[htb]% 
\centering
      {\scriptsize%
        \textsf{%
          \begin{tabular}{|c|c|c|c|}
            \hline
            family size & weight & $\delta$\\  
            \hline
            1  & 7${}^{e}$  & $6.10\cdot 10^{-5}$\\
            2  & 8${}^{e}$  & $3.05\cdot 10^{-5}$\\
            3  & 9${}^{e}$  & $1.14\cdot 10^{-5}$\\
            4  & 10${}^{g}$ & $3.81\cdot 10^{-6}$\\
            6  & 11${}^{g}$ & $1.43\cdot 10^{-6}$\\
            10 & 12${}^{g}$ & $5.97\cdot 10^{-7}$\\
            \hline%
          \end{tabular}
          }%
        }%
      ~~~~~~~ 
      {\scriptsize
        \textsf{ 
          \begin{tabular}{l}
            \SHAPE{\{~\#\#\#\#---\#---------\#---\#--\#\#\#\#~,}\\
            \SHAPE{~~\#\#\#--\#--\#\#--------\#-\#\#\#\#~,}\\
            \SHAPE{~~\#\#\#\#----\#--\#--\#\#-\#\#\#~,}\\
            \SHAPE{~~\#\#\#-\#-\#---\#\#--\#\#\#\#~,}\\
            \SHAPE{~~\#\#\#-\#\#-\#\#--\#-\#-\#\#~,}\\
            \SHAPE{~~\#\#\#\#-\#\#-\#-\#\#\#\#}~\}\\
          \end{tabular}%
          }%
        }
      \caption{\label{seeds32_5} Computed seed families for the
        $(32,5)$-problem and the chosen family (6 seeds of
        weight 11)}%
  \end{figure}% 

The filtering has been applied to a database of rice EST sequences 
  composed of 100015 sequences for a total length of
  42,845,242 bp \footnote{source : {\tt
      http://bioserver.myongji.ac.kr/ricemac.html},   The Korea Rice Genome Database}. 
Substrings matching other substrings with $5$ substitution errors or
less were computed. The computation took slightly more than one hour
on a Pentium\texttrademark  4 3GHz computer. Before applying the filtering using the
family for the $(32,5)$-problem, we made a rough pre-filtering using
one spaced seed of weight $16$ to detect, with a high selectivity,
almost identical regions. 65\% of the database has been discarded by
this pre-filtering. Another 22\% of the database has been filtered out
using the chosen seed family,
leaving the remaining 13\% as oligo candidates. 

\end{subsection}
\end{section}
\begin{section}{Conclusion}
In this paper, we studied a lossless filtration method based on
multi-seed families and demonstrated that it 
represents an improvement compared to the
single-seed approach considered in~\cite{burkhardt03better}. 
We showed how some important characteristics of seed families can be
computed using the dynamic programming. We presented several
combinatorial results that allow one to construct efficient families
composed of seeds with a periodic structure. Finally, we described
a large-scale computational experiment of designing reliable
oligonucleotides for DNA micro-arrays. The obtained experimental results provided evidence of the applicability and
efficiency of the whole method.

The results of Sections~\ref{one-joker}-\ref{regular} establish
several combinatorial properties of seed families, but many more of
them remain to be elucidated.
The structure of optimal or near-optimal seed families
can be reduced to number-theoretic questions, but this
relation remains to be clearly established. 
In general, constructing an algorithm to systematically design seed
families with quality guarantee remains an open problem. 
Some complexity issues remain open too: for example, what is the complexity of
testing if a single seed is lossless for given
$m,k$? Section~\ref{dyn-progr} implies a time bound 
exponential on the number of jokers. 
Note that for multiple seeds, computing the number of detected
similarities is NP-complete
\cite[Section~3.1]{PatternHunter04}. 

Another direction is to consider different distance measures, especially
the Levenstein distance, or at least to allow some restricted
insertion/deletion errors. The method proposed in \cite{BK02} does not
seem to be easily generalized to multi-seed families, and a further work is
required to improve lossless filtering in this case.

{\it Acknowledgements:} G.~Kucherov and L.~No\'e have been supported
by the French {\em Action Sp\'ecifique ``Algorithmes et S\'equences''}
of CNRS. A part of this work has been done during a stay of
M.~Roytberg at LORIA, Nancy, supported by INRIA. 
M.~Roytberg has been supported  by the Russian Foundation for Basic
Research (project nos. 03-04-49469, 02-07-90412) and by grants from
the RF Ministry for Industry, Science, and Technology (20/2002,
5/2003) and NWO.
\end{section}

\bibliographystyle{IEEEtran}
\bibliography{IEEEabrv,paper}

\begin{thebibliography}{10}
\providecommand{\url}[1]{#1}
\csname url@rmstyle\endcsname
\providecommand{\newblock}{\relax}
\providecommand{\bibinfo}[2]{#2}
\providecommand\BIBentrySTDinterwordspacing{\spaceskip=0pt\relax}
\providecommand\BIBentryALTinterwordstretchfactor{4}
\providecommand\BIBentryALTinterwordspacing{\spaceskip=\fontdimen2\font plus
\BIBentryALTinterwordstretchfactor\fontdimen3\font minus
  \fontdimen4\font\relax}
\providecommand\BIBforeignlanguage[2]{{%
\expandafter\ifx\csname l@#1\endcsname\relax
\typeout{** WARNING: IEEEtran.bst: No hyphenation pattern has been}%
\typeout{** loaded for the language `#1'. Using the pattern for}%
\typeout{** the default language instead.}%
\else
\language=\csname l@#1\endcsname
\fi
#2}}

\bibitem{burkhardt03better}
S.~Burkhardt and J.~K\"arkk\"ainen, ``Better filtering with gapped $q$-grams,''
  \emph{Fundamenta Informaticae}, vol.~56, no. 1-2, pp. 51--70, 2003,
  preliminary version in Combinatorial Pattern Matching 2001.

\bibitem{FlexiblePatternMatching02}
G.~Navarro and M.~Raffinot, \emph{Flexible Pattern Matching in Strings --
  Practical on-line search algorithms for texts and biological
  sequences}.\hskip 1em plus 0.5em minus 0.4em\relax Cambridge University
  Press, 2002.

\bibitem{GBLAST97}
S.~Altschul, T.~Madden, A.~Sch\"affer, J.~Zhang, Z.~Zhang, W.~Miller, and
  D.~Lipman, ``Gapped {BLAST }and {PSI-BLAST}: a new generation of protein
  database search programs,'' \emph{Nucleic Acids Research}, vol.~25, no.~17,
  pp. 3389--3402, 1997.

\bibitem{PatternHunter02}
B.~Ma, J.~Tromp, and M.~Li, ``{P}attern{H}unter: Faster and more sensitive
  homology search,'' \emph{Bioinformatics}, vol.~18, no.~3, pp. 440--445, 2002.

\bibitem{BLASTZ03}
S.~Schwartz, J.~Kent, A.~Smit, Z.~Zhang, R.~Baertsch, R.~Hardison, D.~Haussler,
  and W.~Miller, ``Human--mouse alignments with {BLASTZ},'' \emph{Genome
  Research}, vol.~13, pp. 103--107, 2003.

\bibitem{NoeKucherovBMCBioinformatics04}
L.~No{\'e} and G.~Kucherov, ``Improved hit criteria for {DNA} local
  alignment,'' \emph{{BMC} {B}ioinformatics}, vol.~5, no. 149, october 2004.

\bibitem{PW95}
P.~Pevzner and M.~Waterman, ``Multiple filtration and approximate pattern
  matching,'' \emph{Algorithmica}, vol.~13, pp. 135--154, 1995.

\bibitem{FLASH93}
A.~Califano and I.~Rigoutsos, ``Flash: A fast look-up algorithm for string
  homology,'' in \emph{Proceedings of the 1st International Conference on
  Intelligent Systems for Molecular Biology}, July 1993, pp. 56--64.

\bibitem{Buhler02}
J.~Buhler, ``Provably sensitive indexing strategies for biosequence similarity
  search,'' in \emph{Proceedings of the 6th Annual International Conference on
  Computational Molecular Biology (RECOMB02), Washington, DC (USA)}.\hskip 1em
  plus 0.5em minus 0.4em\relax ACM Press, April 2002, pp. 90--99.

\bibitem{KLMT04}
U.~Keich, M.~Li, B.~Ma, and J.~Tromp, ``On spaced seeds for similarity
  search,'' \emph{Discrete Applied Mathematics}, vol. 138, no.~3, pp. 253--263,
  2004.

\bibitem{BKS03}
J.~Buhler, U.~Keich, and Y.~Sun, ``Designing seeds for similarity search in
  genomic {DNA},'' in \emph{Proceedings of the 7th Annual International
  Conference on Computational Molecular Biology (RECOMB03), Berlin
  (Germany)}.\hskip 1em plus 0.5em minus 0.4em\relax ACM Press, April 2003, pp.
  67--75.

\bibitem{BBVb03}
B.~Brejova, D.~Brown, and T.~Vinar, ``Vector seeds: an extension to spaced
  seeds allows substantial improvements in sensitivity and specificity,'' in
  \emph{Proceedings of the 3rd International Workshop in Algorithms in
  Bioinformatics (WABI), Budapest (Hungary)}, ser. Lecture Notes in Computer
  Science, G.~Benson and R.~Page, Eds., vol. 2812.\hskip 1em plus 0.5em minus
  0.4em\relax Springer, September 2003, pp. 39--54.

\bibitem{KucherovNoePontyBIBE04}
G.~Kucherov, L.~No{\'e}, and Y.~Ponty, ``Estimating seed sensitivity on
  homogeneous alignments,'' in \emph{Proceedings of the IEEE 4th Symposium on
  Bioinformatics and Bioengineering (BIBE2004), May 19-21, 2004, Taichung
  (Taiwan)}, ser. the IEEE 4th Symposium on Bioinformatics and Bioengineering -
  BIBE'2004.\hskip 1em plus 0.5em minus 0.4em\relax IEEE Computer Society
  Press, April 2004, pp. 387--394.

\bibitem{ChoiZhang03}
K.~Choi and L.~Zhang, ``Sensitivity analysis and efficient method for
  identifying optimal spaced seeds,'' \emph{Journal of Computer and System
  Sciences}, vol.~68, pp. 22--40, 2004.

\bibitem{MiklosCPM04}
M.~Cs\"{u}r\"{o}s, ``Performing local similarity searches with variable lenght
  seeds,'' in \emph{Proceedings of the 15th Annual Combinatorial Pattern
  Matching Symposium (CPM), Istanbul (Turkey)}, ser. Lecture Notes in Computer
  Science, S.~Sahinalp, S.~Muthukrishnan, and U.~Dogrusoz, Eds., vol.
  3109.\hskip 1em plus 0.5em minus 0.4em\relax Springer Verlag, 2004, pp.
  373--387.

\bibitem{PatternHunter04}
M.~Li, B.~Ma, D.~Kisman, and J.~Tromp, ``{P}attern{H}unter {II}: Highly
  sensitive and fast homology search,'' \emph{Journal of Bioinformatics and
  Computational Biology}, vol.~2, no.~3, pp. 417--440, September 2004.

\bibitem{SB04}
Y.~Sun and J.~Buhler, ``Designing multiple simultaneous seeds for {DNA}
  similarity search,'' in \emph{Proceedings of the 8th Annual International
  Conference on Research in Computational Molecular Biology (RECOMB
  2004)}.\hskip 1em plus 0.5em minus 0.4em\relax ACM Press, March 2004, pp.
  76--84.

\bibitem{Brown2004}
D.~G. Brown, ``Multiple vector seeds for protein alignment,'' in
  \emph{Proceedings of the 4th International Workshop on Algorithms in
  Bioinformatics (WABI), September 2004, Bergen (Norway)}, ser. Lecture Notes
  in Bioinformatics, I.~Jonassen and J.~Kim, Eds., vol. 3240.\hskip 1em plus
  0.5em minus 0.4em\relax Springer Verlag, 2004, pp. 170--181.

\bibitem{XuBrownLiMaCPM04}
J.~Xu, D.~Brown, M.~Li, and B.~Ma, ``Optimizing multiple spaced seeds for
  homology search,'' in \emph{Proceedings of the 15th Symposium on
  Combinatorial Pattern Matching, Istambul (Turkey)}, ser. Lecture Notes in
  Computer Science, S.~Sahinalp, S.~Muthukrishnan, and U.~Dogrusoz, Eds., vol.
  3109, 2004, pp. 47--58.

\bibitem{OommenDong97}
J.~Oommen and J.~Dong, ``Generalized swap-with-parent schemes for
  self-organizing sequential linear lists,'' in \emph{Proceedings of the 1997
  International Symposium on Algorithms and Computation (ISAAC'97), Singapore},
  ser. Lecture Notes in Computer Science, vol. 1350.\hskip 1em plus 0.5em minus
  0.4em\relax Springer, December 17-19 1997, pp. 414--423.

\bibitem{Stormo01}
F.~Li and G.~Stormo, ``Selection of optimal {DNA} oligos for gene expression
  arrays,'' \emph{Bioinformatics}, vol.~17, pp. 1067--1076, 2001.

\bibitem{Kaderali02}
L.~Kaderali and A.~Schliep, ``Selecting signature oligonucleotides to identify
  organisms using {DNA} arrays,'' \emph{Bioinformatics}, vol.~18, no.~10, pp.
  1340--1349, 2002.

\bibitem{Rahmann03}
S.~Rahmann, ``Fast large scale oligonucleotide selection using the longest
  common factor approach,'' \emph{Journal of Bioinformatics and Computational
  Biology}, vol.~1, no.~2, pp. 343--361, 2003.

\bibitem{Zheng03}
J.~Zheng, T.~Close, T.~Jiang, and S.~Lonardi, ``Efficient selection of unique
  and popular oligos for large {EST} databases,'' in \emph{Proceedings of the
  14th Annual Combinatorial Pattern Matching Symposium (CPM), 2003, Morelia
  (Mexico)}, ser. Lecture Notes in Computer Science, vol. 2676.\hskip 1em plus
  0.5em minus 0.4em\relax Springer Verlag, 2003, pp. 273--283.

\bibitem{BK02}
S.~Burkhardt and J.~Karkkainen, ``{O}ne-gapped $q$-gram filters for
  {L}evenshtein {D}istance,'' in \emph{Proceedings of the 13th Symposium on
  Combinatorial Pattern Matching (CPM'02)}, vol. 2373.\hskip 1em plus 0.5em
  minus 0.4em\relax Springer, 2002, pp. 225--234.

\end{thebibliography}
\end{document}